\documentclass[11pt]{article}
\usepackage{amsthm,amsfonts,amsmath,amssymb,times}
\usepackage{enumerate, graphicx}
\usepackage{ xspace}
\usepackage{url}
\usepackage{algo}
\usepackage{hyperref}

\advance \topmargin by -1.0in
\advance \oddsidemargin by -0.8in
\newcommand{\myparskip}{3pt}
\parskip \myparskip

\newcommand{\headers}[3]{
\newpage\setcounter{page}{1}
\def\@oddhead{$\underline{\hbox to\textwidth{%
\textbf{\rlap{#1}\phantom{hj}\hfill #2 \hfill \llap{#3}}}}$}
\def\@oddfoot{\hfill\thepage\hfill}}

\newtheorem*{theorem*}{Theorem}

\newtheorem{lemma}{Lemma}[section]
\newtheorem{theorem}[lemma]{Theorem}

\newtheorem{prop}[lemma]{Proposition}
\newtheorem*{claim}{Claim}

\newtheorem{prob}{Problem}

\renewenvironment{proof}{\vspace{-0.1in}\noindent{\bf Proof:}}%
        {\hspace*{\fill}$\Box$\par}
        {\hspace*{\fill}$\Box$\par}
        {\hspace*{\fill}$\Box$\par}

\def\bar{\overline}

\def\script#1{\mathcal{#1}}

\def\mL{\script{L}}
\def\mE{\script{E}}

\newcommand{\sndp}{\text{SNDP}\xspace}
\newcommand{\ecsndp}{\text{EC-SNDP}\xspace}
\newcommand{\elemsndp}{\text{Elem-SNDP}\xspace}
\newcommand{\vcsndp}{\text{VC-SNDP}\xspace}
\newcommand{\hypersndp}{\text{Hypergraph-SNDP}\xspace}

\renewcommand{\a}{\alpha}
\renewcommand{\b}{\beta}
\newcommand{\dplus}{\ensuremath{d^+}}
\newcommand{\rmax}{r_{\max}}

\title{A note on the Survivable Network Design Problem}
\author{
Chandra Chekuri\thanks{Dept.\ of Computer Science, Univ.\ of Illinois at Urbana-Champaign, Urbana, IL, USA. {\tt chekuri@illinois.edu}. Supported in
part by NSF grants CCF-1016684, CCF-1319376 and CCF-1526799.}
\and
Thapanapong Rukkanchanunt\thanks{Chiang Mai University, Chiang Mai, Thailand.
{\tt thapanapong.r@cmu.ac.th}. The author was an undergraduate student at Univ.\ of Illinois when he worked on results of this paper.}
}
\begin{document}
\maketitle

\begin{abstract}
  In this note we consider the survivable network design problem
  (\sndp) in undirected graphs. We make two contributions.  The first
  is a new counting argument in the iterated rounding based
  $2$-approximation for edge-connectivity \sndp (\ecsndp) originally
  due to Jain~\cite{Jain01}. The second is to make some additional
  connections between hypergraphic version of \sndp (\hypersndp)
  introduced in \cite{ZhaoNI03} and edge and node-weighted versions of
  \ecsndp and element-connectivity \sndp (\elemsndp). One useful
  consequence of this connection is a $2$-approximation for \elemsndp
  that avoids the use of set-pair based relaxation and analysis.
\end{abstract}

\section{Introduction}
\label{sec:intro}
The {\em survivable network design problem} (\sndp) is a fundamental
problem in network design and has been instrumental in the development
of several algorithmic techniques. The input to \sndp is a graph
$G=(V,E)$ and an integer requirement $r(uv)$ between each unordered
pair of nodes $uv$.  The goal is to find a minimum-cost subgraph $H$
of $G$ such that for each pair $uv$, the connectivity in $H$ between
$u$ and $v$ is at least $r(uv)$. We use $\rmax$ to denote $\max_{uv}
r(uv)$, the maximum requirement.  We restrict attention to undirected
graphs in this paper.  There are several variants depending on whether
the costs are on edges or on nodes, and whether the connectivity
requirement is edge, element or node connectivity. Unless otherwise
specified we will assume that $G$ has edge-weights $c:E \rightarrow
\mathbb{R}_+$.  We refer to the three variants of interest based on
edge, element and vertex connectivity as \ecsndp, \elemsndp and
\vcsndp. All of them are NP-Hard and APX-hard to approximate even in
very special cases.

The seminal work of Jain \cite{Jain01} obtained a $2$-approximation
for \ecsndp via the technique of iterated rounding that was introduced
in the same paper.  A $2$-approximation for \elemsndp was obtained,
also via iterated rounding, in \cite{FleischerJW06,CheriyanVV06}. For
\vcsndp the current best approximation bound is $O(\rmax^3\log |V|)$
\cite{ChuzhoyK}; it is also known from hardness results in
\cite{ChakrabortyCK08} that the approximation bound for \vcsndp must
depend polynomially on $\rmax$ under standard hardness assumptions.

In this note we revisit the iterated rounding framework that yields a
$2$-approximation for \ecsndp and \elemsndp. The framework is based on
arguing that for a class of covering problems, a basic feasible
solution to an LP relaxation for the covering problem has a variable
of value at least $\frac{1}{2}$. This variable is then rounded up to
$1$ and the residual problem is solved inductively. A key fact needed
to make this iterative approach work is that the residual problem lies
in the same class of covering problems. This is ensured by working
with the class of skew-supermodular (also called weakly-supermodular)
requirement functions which capture \ecsndp as a special case. The
proof of existence of an edge with large value in a basic feasible
solution for this class of requirement functions has two components.
The first is to establish that a basic feasible solution is
characterized by a laminar family of sets in the case of \ecsndp (and
set pairs in the case of \elemsndp). The second is a counting argument
that uses this characterization to obtain a contradiction if no
variable is at least $\frac{1}{2}$. The counting argument of Jain
\cite{Jain01} has been simplified and streamlined in subsequent work
via fractional token arguments \cite{BansalKN09,NagarajanRS10}.  These
arguments have been applied for several related problems for which
iterated rounding has been shown to be a powerful technique; see
\cite{LauRS-book}.  The fractional token argument leads to short and
slick proofs. At the same time we feel that it is hard to see the
intuition behind the argument. Partly motivated by pedgogical
reasons, in this note, we revisit the counting argument for \ecsndp
and provide a different counting argument along with a longer
explanation. The goal is to give a more combinatorial flavor to the
argument. We give this argument in Section~\ref{sec:counting}.

The second part of the note is on \elemsndp. A $2$-approximation for
this problem has been derived by generalizing the iterated rounding
framework to a set-pair based relaxation
\cite{FleischerJW06,CheriyanVV06}. The set-pair based relaxation and
arguments add substantial notation to the proofs although one can see
that there are strong similarities to the basic argument used in
\ecsndp. The notational overhead limits the ability to teach and
understand the proof for \elemsndp. Interestingly, in a little noticed
paper, Zhao, Nagamochi and Ibaraki \cite{ZhaoNI03} defined a
generalization of \ecsndp to hypergraphs which we refer to as
\hypersndp. They observed that \elemsndp can be easily reduced to
\hypersndp in which the only non-zero weight hyperedges are of size
$2$ (regular edges in a graph). The advantage of this reduction is
that one can derive a $2$-approximation for \elemsndp by essentially
appealing to the same argument as for \ecsndp with a few minor
details. We believe that this is a useful perspective.  Second, there
is a simple and well-known connection between node-weighted network
design in graphs and network design problems on hypergraphs. We
explicitly point these connections which allows us to derive some
results for \hypersndp.  Section~\ref{sec:hypergraph} describes these
connections and results.

This note assumes that the reader has some basic familiarity with
previous literature on \sndp and iterated rounding.

\section{Iterated rounding for \ecsndp}
\label{sec:counting}

The $2$-approximation for \ecsndp is based on casting it as a special
case of covering a skew-supermodular requirement function by a graph.
We set up the background now. Given a finite ground set $V$ 
an integer valued set function $f: 2^V \rightarrow \mathbb{Z}$ is
skew-supermodular if for all $A, B \subseteq V$ 
one of the following holds:
\begin{eqnarray*}
  f(A) + f(B) & \le &f(A \cap B) + f(A \cup B) \\
f(A) + f(B) & \le & f(A - B) + f(B-A)
\end{eqnarray*}
Given an edge-weighted graph $G=(V,E)$ and a skew-supermodular
requirement function $f: 2^V \rightarrow \mathbb{Z}$, we can consider
the problem of finding the minimum-cost subgraph $H=(V,F)$ of $G$ such
that $H$ covers $f$; that is, for all $S \subseteq V$, $|\delta_F(S)|
\ge f(S)$.  Here $\delta_F(S)$ is the set of all edges in $F$ with one
endpoint in $S$ and the other outside.  Given an instance of \ecsndp
with input graph $G=(V,E)$ and edge-connectivity requirements $r(uv)$
for each pair $uv$, we can model it by setting $f(S) = \max_{u \in S,
  v \not \in S} r(uv)$. It can be verified that $f$ is
skew-supermodular. The important aspect of skew-supermodular functions
that make them well-suited for the iterated rounding approach is the
following.

\begin{lemma} [\cite{Jain01}]
  \label{lem:skew-minus-cut}
  Let $G=(V,E)$ be a graph and $f: 2^V \rightarrow \mathbb{Z}$ be a
  skew-supermodular requirement function, and $F \subseteq E$ be a
  subset of edges. The residual requirement function $g:2^V
  \rightarrow \mathbb{Z}$ defined by $g(S) = f(S) - |\delta_F(S)|$ for
  each $S \subseteq V$ is also skew-supermodular.
\end{lemma}

Although the proof is standard by now we will state it in a more
general way.

\begin{lemma}
  \label{lem:skew-minus-submod}
  Let $f: 2^V \rightarrow \mathbb{Z}$ be a skew-supermodular
  requirement function and let $h : 2^V \rightarrow \mathbb{Z}_+$ be a
  symmetric submodular function. Then $g = f - h$ is a
  skew-supermodular function.
\end{lemma}

\begin{proof}
  Since $h$ is submodular we have that for all $A, B \subseteq V$,
  $$ h(A) + h(B) \ge h(A \cup B) + h(A \cap B).$$
  Since $h$ is also symmetric it is posi-modular which means that for all
  $A, B \subseteq V$,
  $$ h(A) + h(B) \ge h(A - B) + h(B - A).$$
  Note that $h$ satisfies both properties for each $A,B$. It is now easy
  to check that $f-h$ is skew-supermodular. 
\end{proof}

Lemma~\ref{lem:skew-minus-cut} follows from
Lemma~\ref{lem:skew-minus-submod} by noting that the cut-capacity
function $|\delta_F|: 2^V \rightarrow \mathbb{Z}_+$ is submodular and
symmetric in undirected graphs. We also note that the same property holds
for the more general setting when $G$ is a hypergraph.

The standard LP relaxation for covering a function by a graph is described
below where there is variable $x_e \in [0,1]$ for each edge $e \in E$.
\begin{eqnarray*}
  \min \sum_{e \in E} c_e x_e &&\\
  \sum_{e \in \delta(S)} x_e & \ge & f(S) \quad \quad S \subset V\\
  x_e & \in & [0,1] \quad \quad e \in E
\end{eqnarray*}

The technical theorem that underlies the $2$-approximation for \ecsndp
is the following.
\begin{theorem} [\cite{Jain01}] 
\label{thm:main-ecsndp}
Let $f$ be a non-trivial\footnote{We use the term non-trivial to
  indicate that there is at least one set $S \subset V$ such that
  $f(S) > 0$.} skew-supermodular function. In any basic feasible
solution $\bar{x}$ to the LP relaxation of covering $f$ by a graph $G$
there is an edge $e$ such that $\bar{x}_e \ge \frac{1}{2}$.
\end{theorem}

To prove the preceding theorem it suffices to focus on basic feasible
solutions $\bar{x}$ that are fully fractional; that is, $\bar{x}_e \in
(0,1)$ for all $e$. For a set of edges $F \subseteq E$ let $\chi(F)
\in \{0,1\}^{|E|}$ denote the characteristic vector of $F$; that is, a
$|E|$-dimensional vector that has a $1$ in each position corresponding
to an edge $e \in F$ and a $0$ in all other positions.
Theorem~\ref{thm:main-ecsndp} is built upon the following
characterization of basic feasible solutions and is shown via
uncrossing arguments.

\begin{lemma}[\cite{Jain01}]
  \label{lem:laminar}
  Let $\bar{x}$ be a fully-fractional basic feasible solution to the
  the LP relaxation. Then there is a laminar family of vertex subsets
  $\mL$ such that $\bar{x}$ is the unique solution to the system of
  equalities 
    $$x(\delta(S)) = f(S) \quad \quad S \in \mL.$$
    In particular this also implies that $|\mL| = |E|$ and
    that the vectors $\chi(\delta(S))$, $S \in \mL$ are 
    linearly independent.
\end{lemma}

The second part of the proof of Theorem~\ref{thm:main-ecsndp} 
is a counting argument that relies on the characterization in
Lemma~\ref{lem:laminar}. The rest of this section describes
a counting argument which we believe is slightly different from the
previous ones in terms of the main invariant. The goal is to derive
it organically from simpler cases.

With every laminar family we can associate a rooted forest.  We use
terminology for rooted forests such as leaves and roots as well as set
terminology.  We refer to a set $C \in \mL$ as a child of a set $S$ if
$C \subset S$ and there is no $S' \in \mL$ such that $C \subset S'
\subset S$; If $C$ is the child of $S$ then $S$ is the parent of $C$.
Maximal sets of $\mL$ correspond to the roots of the forest associated
with $\mL$.

\subsection{Counting Argument}
The proof is via contradiction where we assume that
$0 < \bar{x}_e < \frac{1}{2}$ for each $e \in E$.
We call the two nodes incident to an edge as the endpoints of the
edges. We say that an endpoint $u$ {\em belongs} to a set $S \in \mL$
if $u$ is the minimal set from $\mL$ that contains $u$. 

We consider the simplest setting where $\mL$ is a collection of
disjoint sets, in other words, all sets are maximal. In this case the
counting argument is easy. Let $m = |E| = |\mL|$.  For each $S \in
\mL$, $f(S) \ge 1$ and $\bar{x}(\delta(S)) = f(S)$.  If we assume that
$\bar{x}_e < \frac{1}{2}$ for each $e$, we have $|\delta(S)| \ge 3$
which implies that each $S$ contains at least $3$ distinct endpoints. 
Thus, the $m$ disjoint sets require a total of $3m$ endpoints.  However the
total number of endpoints is at most $2m$ since there are $m$ edges,
leading to a contradiction.

Now we consider a second setting where the forest associated with
$\mL$ has $k$ leaves and $h$ internal nodes but each internal node
has at least two children. In this case, following Jain, we can easily
prove a weaker statement that $\bar{x}_e \ge 1/3$ for some edge $e$.
If not, then each leaf set $S$ must have four edges leaving it and hence
the total number of endpoints must be at least $4k$. However, if
each internal node has at least two children, we have $h < k$ and
since $h+k = m$ we have $k > m/2$. This implies that there must
be at least $4k > 2m$ endpoints since the leaf sets are disjoint.
But $m$ edges can have at most $2m$ endpoints.
Our assumption on each internal node having at least two children
is obviously a restriction. So far we have not used the fact that
the vectors $\chi(\delta(S)), S \in \mL$ are linearly independent. We can 
handle the general case to prove $\bar{x}_e \ge 1/3$ by using the
following lemma.

\begin{lemma}[\cite{Jain01}]
  Suppose $C$ is a unique child of $S$. Then there must be at least 
  two endpoints in $S$ that belong to $S$.
\end{lemma}
\begin{proof}
  If there is no endpoint that belongs to $S$ then $\delta(S) =
  \delta(C)$ but then $\chi(\delta(S))$ and $\chi(\delta(C))$ are
  linearly dependent.  Suppose there is exactly one endpoint that
  belongs to $S$ and let it be the endpoint of edge $e$. But then
  $\bar{x}(\delta(S)) = \bar{x}(\delta(C)) + \bar{x}_e$ or
  $\bar{x}(\delta(S)) = \bar{x}(\delta(C)) - \bar{x}_e$.  Both cases
  are not possible because $\bar{x}(\delta(S)) = f(S)$ and
  $\bar{x}(\delta(C)) = f(C)$ where $f(S)$ and $f(C)$ are positive
  integers while $\bar{x}_e \in (0,1)$. Thus there are at least two
  end points that belong to $S$.
\end{proof}

Using the preceding lemma we prove that $\bar{x}_e \ge 1/3$ for some
edge $e$.  Let $k$ be the number of leaves in $\mL$ and $h$ be the
number of internal nodes with at least two children and let $\ell$ be
the number of internal nodes with exactly one child. We again have $h
< k$ and we also have $k + h + \ell = m$. Each leaf has at least four
endpoints. Each internal node with exactly one child has at least two
end points which means the total number of endpoints is at least $4k +
2\ell$. But $4k + 2\ell = 2k + 2k + 2\ell > 2k + 2h + 2\ell > 2m$ and
there are only $2m$ endpoints for $m$ edges. In other words, we can
ignore the internal nodes with exactly one child since there are two
endpoints in such a node/set and we can effectively charge one edge to
such a node.

We now come to the more delicate argument to prove the tight bound
that $\bar{x}_e \ge \frac{1}{2}$ for some edge $e$. Our main contribution is
to show an invariant that effectively reduces the argument to the case
where we can assume that $\mL$ is a collection of leaves. This is
encapsulated in the claim below which requires some notation. Let $\a(S)$ be
the number of sets of $\mL$ contained in $S$ including $S$ itself. Let
$\b(S)$ be the number of edges whose {\em both} endpoints lie inside
$S$.  Recall that $f(S)$ is the requirement of $S$.

\begin{claim}
  For all $S \in \mL$, $f(S) \geq \a(S) - \b(S)$.
\end{claim}

Assuming that the claim is true we can do an easy counting argument.
Let $R_1,R_2,\ldots,R_h$ be the maximal sets in $\mL$ (the roots of the
forest). Note that $\sum_{i=1}^h \a(R_i) = |\mL| = m$.
Applying the claim to each $R_i$ and summing up,
$$\sum_{i=1}^h f(R_i) \ge \sum_{i=1}^h \a(R_i) - \sum_{i=1}^h \b(R_i) \ge m -  \sum_{i=1}^h\b(R_i).$$

Note that $\sum_{i=1}^h f(R_i)$ is the total requirement of the 
maximal sets. And $m -  \sum_{i=1}^h\b(R_i)$ is the total number of
edges that cross the sets $R_1,\ldots,R_h$. Let $E'$ be the
set of edges crossing these maximal sets.  Now we are back to 
the setting with $h$ disjoint sets and $E'$ edges
with $\sum_{i=1}^h f(R_i) \ge |E'|$. This easily leads to a contradiction
as before if we assume that $\bar{x}_e < \frac{1}{2}$ for all $e \in E'$.
Formally, each set $R_i$ requires $> 2f(R_i)$ edges crossing it from $E'$
and therefore $R_i$ contains at least $2f(R_i)+1$ endpoints of edges
from $E'$. Since
$R_1,\ldots,R_h$ are disjoint the total number of endpoints is
at least $2\sum_i f(R_i) + h$ which is strictly more than $2|E'|$.

Thus, it remains to prove the claim which we do by inductively starting
at the leaves of the forest for $\mL$. 

\medskip
\noindent
{\bf Case 1:} $S$ is a leaf node. 
We have $f(S) \geq 1$ while $\a(S) = 1$ and $\b(S) = 0$ which verifies the
claim.

\begin{figure}[tb]
  \centering
  \includegraphics[width=3.5in]{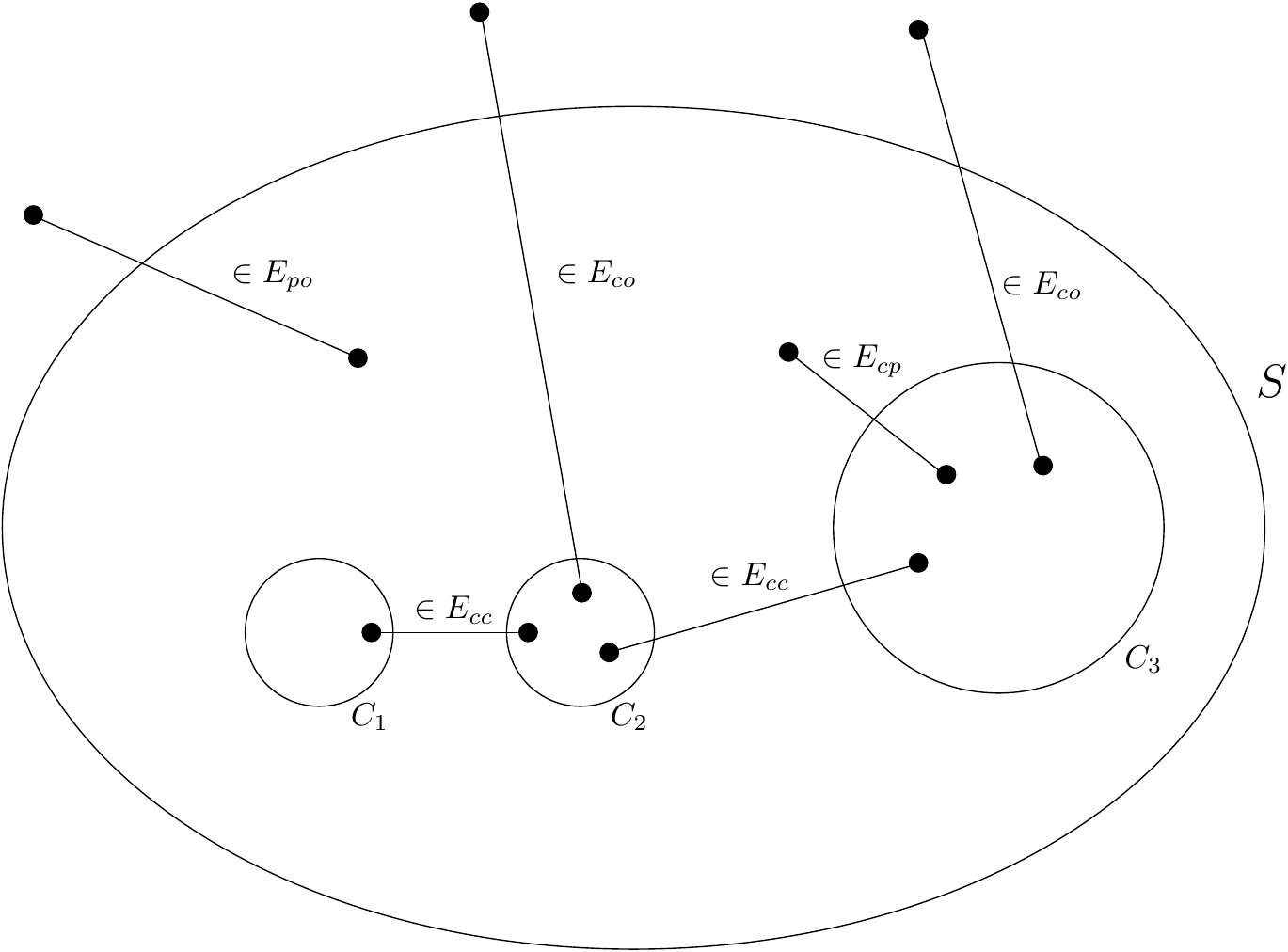}
  \caption{$S$ is an internal node with several children. Different
  types of edges that  play a role. $p$ refers to parent set $S$,
$c$ refer to a child set and $o$ refers to outside.}
  \label{fig:counting}
\end{figure}

\medskip
\noindent {\bf Case 2:} $S$ is an internal nodes with $k$ children
$C_1,C_2,\ldots,C_k$. See Fig~\ref{fig:counting} for the different
types of edges that are relevant. $E_{cc}$ is the set
of edges with end points in two different children of $S$. $E_{cp}$
be the set of edges that cross exactly one child but do not cross
$S$. $E_{po}$ be the set of edges that cross $S$ but do not cross
any of the children. $E_{co}$ is the set of edges that cross both
a child and $S$.  This notation is borrowed from \cite{WilliamsonSbook}.

Let $\gamma(S)$ be the number of edges 
whose both endpoints belong to $S$ but not to any child of $S$. 
Note that $\gamma(S) = |E_{cc}| + |E_{cp}|$.

Then,
\begin{eqnarray}
	\b(S)	&=&	\gamma(S) + \sum_{i=1}^k \b(C_i) \nonumber \\
		&\geq&	\gamma(S) + \sum_{i=1}^k \a(C_i) - \sum_{i=1}^k f(C_i) \label{eqn:induction}\\
		&=&	\gamma(S) + \a(S) - 1 - \sum_{i=1}^k f(C_i) \nonumber
\end{eqnarray}
(\ref{eqn:induction}) follows by applying the inductive hypothesis to
each child. From the preceding inequality, to prove that $\b(S) \geq
\a(S) - f(S)$ (the claim for $S$), it suffices to show the following
inequality.
\begin{eqnarray}
	\gamma(S)	&\geq&	\sum_{i=1}^k f(C_i) - f(S) + 1.
\end{eqnarray}

The right hand side of the above inequality can be written as:
\begin{equation}
  \label{eq:diff}
  \sum_{i=1}^k f(C_i) - f(S) + 1	= \sum_{e \in E_{cc}} 2x_e + \sum_{e \in E_{cp}} x_e - \sum_{e \in E_{po}} x_e + 1.
\end{equation}

We consider two subcases.

\smallskip
\noindent {\bf Case 2.1:} $\gamma(S) = 0$. This implies that $E_{cc}$
and $E_{cp}$ are empty. Since $\chi(\delta(S))$ is linearly
independent from $\chi(\delta(C_1)),\ldots,\chi(\delta(C_k))$, we must
have that $E_{po}$ is not empty and hence $\sum_{e \in E_{po}} x_e >
0$. Therefore, in this case,
$$\sum_{i=1}^k f(C_i) - f(S) + 1	= \sum_{e \in E_{cc}} 2x_e + \sum_{e \in E_{cp}} x_e - \sum_{e \in E_{po}} x_e + 1 = - \sum_{e \in E_{po}} x_e + 1 < 1.$$
Since the left hand side is an integer, it follows that
$\sum_{i=1}^k f(C_i) - f(S) + 1	\le 0 = \gamma(S)$.

\smallskip
\noindent {\bf Case 2.2:} $\gamma(S) \geq 1$.  Recall 
that $\gamma(S) = |E_{cc}| + |E_{cp}|$.
$$\sum_{i=1}^k f(C_i) - f(S) + 1 = \sum_{e \in E_{cc}} 2x_e + \sum_{e \in E_{cp}} x_e - \sum_{e \in E_{po}} x_e + 1  \le   \sum_{e \in E_{cc}} 2x_e + \sum_{e \in E_{cp}} x_e  + 1$$
By our assumption that $\bar{x}_e < \frac{1}{2}$ for each $e$, we have
$\sum_{e \in E_{cc}}2x_e < |E_{cc}|$ if $|E_{cc}| > 0$, and similarly
$\sum_{e \in E_{cp}} x_e < |E_{cp}|/2$ if $|E_{cp}| > 0$. Since
$\gamma(S) = |E_{cc}| + |E_{cp}| \ge 1$ we conclude that
$$ \sum_{e \in E_{cc}} 2x_e + \sum_{e \in E_{cp}} x_e  < \gamma(S).$$
Putting together we have 
$$\sum_{i=1}^k f(C_i) - f(S) + 1 \le  \sum_{e \in E_{cc}} 2x_e + \sum_{e \in E_{cp}} x_e  + 1  < \gamma(S) + 1 \le \gamma(S)$$
as desired.

This completes the proof of the claim.

\section{Connections between \hypersndp, \ecsndp and \elemsndp}
\label{sec:hypergraph}
Zhao, Nagamochi and Ibaraki \cite{ZhaoNI03} considered the extensions
\ecsndp to hypergraphs. In an hypergraph $G=(V,\mE)$ each edge $e \in
\mE$ is a subset of $V$. The degree $d$ of a hypergraph is $\max_{e
  \in \mE} |e|$. Graphs are hypergraphs of degree $2$.  Given a set of
hyperedges $F \subseteq \mE$ and a vertex subset $S \subset V$, we use
$\delta_F(S)$ to denote the set all of all hyperedges in $F$ that have
at least one endpoint in $S$ and at least one endpoint in $V \setminus
S$. It is well-known that $|\delta_F|: 2^V \rightarrow \mathbb{Z}_+$
is a symmetric submodular function.

\hypersndp is defined as follows. The input consists of an
edge-weighted {\em hypergraph} $G=(V,\mE)$ and integer requirements
$r(uv)$ for each vertex pair $uv$. The goal is to find a minimum-cost
hypergraph $H=(V, \mE')$ with $\mE' \subseteq \mE$ such that for all
$uv$ and all $S$ that separate $u,v$ (that is $|S \cap \{u,v\}| = 1$),
we have $|\delta_{E'}(S)| \ge r(uv)$. \hypersndp is a special case of
covering a skew-supermodular requirement function by a hypergraph.  It
is clear that \hypersndp generalizes \ecsndp. Interestingly,
\cite{ZhaoNI03} observed, via a simple reduction, that \hypersndp
generalizes \elemsndp as well.  We now describe \elemsndp formally and
briefly sketch the reduction from \cite{ZhaoNI03}, and subsequently
describe some implications of this connection.

In \elemsndp the input consists of an undirected edge-weighted graph
$G=(V, E)$ with $V$ partitioned into terminals $T$ and non-terminals
$N$. The ``elements'' are the edges and non-terimals, $N \cup E$.  For
each pair $uv$ of terminals there is an integer requirement $r(uv)$,
and the goal is to find a min-cost subgraph $H$ of $G$ such that for
each pair $uv$ of terminals there are $r(uv)$ element-disjoint paths
from $u$ to $v$ in $H$. Note that element-disjoint paths can intersect
in terminals. The notion of element-connectivity and \elemsndp have
been useful in several settings in generalizing edge-connectivity problems
while having some feastures of vertex connectivity. In particular, 
the current approximation for \vcsndp relies on \elemsndp \cite{ChuzhoyK}.

The reduction of \cite{ZhaoNI03} from \elemsndp to \hypersndp is quite
simple. It basically replaces each non-terminal $u \in N$ by a hyperedge. 
The reduction is depicted in the figure below.

\begin{figure}[htb]
  \centering
  \includegraphics[width=2in]{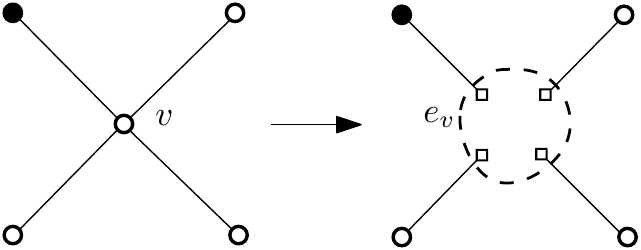}
  \caption{Reducing \elemsndp to \hypersndp. Each non-terminal $v$ is replaced
    by a hyperedge $e_v$ by introducing dummy vertices on each edge
    incident to $v$. The original edges retain their cost while the new
  hyperedges are assigned a cost of zero.}
  \label{fig:elem-to-hyper}
\end{figure}
The reduction shows that an instance of \elemsndp on $G$ can be
reduced to an instance of \hypersndp on a hypergraph $G'$ where the
only hyperedges with non-zero weights in $G'$ are the edges of the
graph $G$. This motivates the definition of $\dplus(G)$ 
which is the maximum degree of a hyperedge in $G$ that has non-zero cost. 
Thus \elemsndp reduces to instances of \hypersndp with $\dplus = 2$.
In fact we can see that the same reduction proves the following.

\begin{prop}
  Node-weighted \elemsndp in which weights are only on non-terminals
  can be reduced in an approximation preserving fashion to
  \hypersndp. In this reduction $\dplus$ of the resulting instance of
  \hypersndp is equal to $\Delta$, the maximum degree of a
  non-terminal with non-zero weight in the instance of node-weighted
  \elemsndp.
\end{prop}

\paragraph{Reducing \elemsndp to problem of covering skew-supermodular
  functions by graphs:}
We saw that an instance of \elemsndp on a graph $H$ can be reduced to
an instance of \hypersndp on a graph $G$ where $\dplus(G) = 2$.
\hypersndp on $G=(V,\mE)$ corresponds to covering a skew-supermodular
function $f:2^V \rightarrow \mathbb{Z}$ by $G$. Let $\mE = F \uplus
\mE'$ where $\mE'$ is the set of all hyperedges in $G$ with degree
more than $2$; thus $F$ is the set of all hyperedges of degree $2$ and
hence $(V,F)$ is a graph. Since each edge in $\mE'$ has zero cost we
can include all of them in our solution, and work with the residual requirement
function $g = f - |\delta_{\mE'}|$. From
Lemma~\ref{lem:skew-minus-submod} and the fact that the cut-capacity
function of a hypergraph is also symmetric and submodular, $g$ is a
skew-supermodular function.  Thus covering $f$ by a min-cost
sub-hypergraph of $G$ can be reduced to covering $g$ by a min-cost
sub-graph of $G'=(V,F)$. We have already seen a $2$-approximation 
for this in the context of \ecsndp. The only issue is whether there is an
efficient separation oracle for solving the LP for covering $g$ by
$G'$. This is a relatively easy exercise using flow arguments and we
omit them. The main point we wish to make is that this reduction
avoids working with set-pairs that are typically used for \elemsndp.
It is quite conceivable that the authors of \cite{ZhaoNI03} were aware
of this simple connection but it does not seem to have been made explicitly
in their paper or in \cite{ZhaoNI02}.

\paragraph{Approximating  \hypersndp:} \cite{ZhaoNI03} derived a $\dplus H_{\rmax}$ approximationf for \hypersndp where $H_k = 1 + 1/2+ \ldots + 1/k$ is the $k$'th
harmonic number. They obtain this bound via the augmentation framework
for network design \cite{GoemansGPSTW94} and a primal-dual algorithm
in each stage. In \cite{ZhaoNI02} they also observe that \hypersndp
can be reduced to \elemsndp via the following simple reduction. Given
a hypergraph $G=(V,\mE)$ let $H = (V \cup N, E)$ be the standard
bipartite graph representation of $G$ where for each hyperedge $e \in
\mE$ there is a node $z_e \in N$; $z_e$ is connected by edges in $H$
to each vertex $a \in e$. Let $r(uv)$ be the hyperedge connectivity
requirement between a pair of vertices $uv$ in the original instance of
\hypersndp. In $H$ we label $V$ as terminals and $N$ as non-terminals.
For any pair of vertices $uv$ with $u,v \in V$, it is not hard to verify
that the element-connectivity betwee $u$ and $v$ in $H$ is the
same as the hyperedge connectivity in $G$. See \cite{ZhaoNI02} for
details. It remains to model the costs such that an
approximation algorithm for element-connectivity in $H$ can be
translated into an approximation algorithm for hyperedge connectivity
in $G$. This is straightforward. We simply assign cost to non-terminals
in $H$; that is each node $z_e \in N$ corresponding to a hyperedge $e \in \mE$
is assigned a cost equal to $c_e$. We obtain the following easy 
corollary.

\begin{prop}
  \hypersndp can be reduced to node-weighted \elemsndp in an approximation
  preservation fashion.
\end{prop}

\cite{ZhaoNI02} do not explicitly mention the above but note that one
can reduce \hypersndp to (edge-weighted) \elemsndp as follows. Instead
of placing a weight of $c_e$ on the node $z_e$ corresponding to the
hyperedge $e \in \mE$, they place a weight of $c_e/2$ on each edge
incident to $z_e$. This transformation loses an
approximation ratio of $\dplus(G)/2$. From this they conclude that a
$\beta$-approximation for \elemsndp implies a $\dplus \beta/2$-approximation
for \hypersndp; via the $2$-approximation for \elemsndp we obtain
a \dplus approximation for \hypersndp. One can view this as reducing
a node-weighted problem to an edge-weighted problem by transferring
the cost on the nodes to all the edges incident to the node. Since
a non-terminal can only be useful if it has at least two edges incident
to it, in this particular case, we can put a weight of half the node
on the edges incident to the node. A natural question here is whether
one can directly get a $\dplus$ approximation for \hypersndp without
the reduction to \elemsndp. We raise the following technical question.

\begin{prob}
  Suppose $f$ is a non-trivial skew-supermodular function on $V$
  and $G=(V,\mE)$ be a hypergraph. Let $\bar{x}$ be a basic feasible
  solution to the LP for covering $f$ by $G$. Is there an hyperedge
  $e \in \mE$ such that $\bar{x}_e \ge \frac{1}{d}$ where $d$ is the
  degree of $G$?
\end{prob}

The preceding propositions show that \hypersndp is essentially
equivalent to node-weighted \elemsndp where the node-weights are only
on non-terminals. Node-weighted Steiner tree can be reduced to
node-weighted \elemsndp and it is known that Set Cover reduces in an
approximation preserving fashion to node-weighted Steiner tree
\cite{KleinR95}.  Hence, unless $P = NP$, we do not expect a better
than $O(\log n)$-approximation for \hypersndp where $n = |V|$ is
the number of nodes in the graph. Thus, the approximation ratio for
\hypersndp cannot be a constant independent of \dplus. Node-weighted
\elemsndp admits an $O(\rmax \log |V|)$ approximation; see
\cite{Nutov09,ChekuriEV12a,ChekuriEV12b,Fukunaga15}. For planar graphs,
and more generally graphs from a proper minor-closed family,
an improved bound of $O(\rmax)$ is claimed in \cite{ChekuriEV12a}.  The
$O(\rmax \log |V|)$ bound can be better than the bound of $\dplus$ in some
instances. Here we raise a question based on the fact that planar
graphs have constant average degree which is used in the analyis
for node-weighted network design.

\begin{prob}
  Is there an $O(1)$-approximation for node-weighted \ecsndp and
  \elemsndp in planar graphs, in particular when $\rmax$ is a fixed
  constant?
\end{prob}

Finally, we hope that the counting argument and the connections
between \hypersndp, \ecsndp and \elemsndp will be useful for related
problems including the problems involving degree constraints in
network design.  

\bibliographystyle{alpha} \bibliography{sndp}
\end{document}